\documentclass[a4paper]{article}

\usepackage[document]{ragged2e}

\usepackage[top=1.5in,left=1.5in,right= 1in, bottom = 1.5 in]{geometry}

\usepackage{amsmath,amssymb}
\usepackage{amsfonts}
\usepackage{amsthm}
\usepackage{epstopdf}

\usepackage{enumerate}
\usepackage{yfonts}
\usepackage{mathrsfs}
\usepackage{graphicx}

\usepackage{multirow}
\usepackage[table,xcdraw]{xcolor}

\newtheorem{remark}{Remark}
\newtheorem{definition}{Definition}
\newtheorem{prop}{Proposition}

\usepackage{algorithm}
\usepackage{algorithmic}

\usepackage{changepage}

\usepackage[utf8x]{inputenc}

\usepackage{textcomp,marvosym}

\usepackage{cite}

\usepackage{nameref,hyperref}

\usepackage[right]{lineno}

\usepackage{microtype}
\DisableLigatures[f]{encoding = *, family = * }

\usepackage[table]{xcolor}

\usepackage{array}

\newcolumntype{+}{!{\vrule width 2pt}}

\newlength\savedwidth

\newcommand\thickhline{\noalign{\global\savedwidth\arrayrulewidth\global\arrayrulewidth 2pt}%
\hline
\noalign{\global\arrayrulewidth\savedwidth}}

\raggedright
\usepackage[aboveskip=1pt,labelfont=bf,labelsep=period,justification=raggedright,singlelinecheck=off]{caption}

\usepackage{scrtime}

\title{Robust policy evaluation from large-scale observational studies}

\author{Md Saiful Islam \thanks{Department of Mechanical $\&$ Industrial Engineering, Northeastern University, Boston, MA 02115, USA 
 }
 \and Md Sarowar Morshed \footnotemark[1]
\and Gary J. Young \footnote{Center for Health Policy and Healthcare Research, Northeastern University, Boston, Massachusetts, USA}
\footnote{D'Amore-McKim School of Business, Northeastern University, Boston, Massachusetts, USA}
\thanks{Bouv\`{e} College of Health Sciences, Northeastern University, Boston, Massachusetts, USA
  }
\and Md. Noor-E-Alam \footnotemark[1] \footnotemark[2] \thanks{Corresponding Author: mnalam@neu.edu}}

\begin{document}
\maketitle

\justify

\section*{Abstract}

Under current policy decision making paradigm, we make or evaluate a policy decision by intervening different socio-economic parameters and analyzing the impact of those interventions. This process involves identifying the causal relation between interventions and outcomes. Matching method is one of the popular techniques to identify such causal relations. However, in one-to-one matching, when a treatment or control unit has multiple pair assignment options with similar match quality, different matching algorithms often assign different pairs. Since, all the matching algorithms assign pair without considering the outcomes, it is possible that with same data and same hypothesis, different experimenters can make different conclusions. This problem becomes more prominent in case of large-scale observational studies. Recently, a robust approach is proposed to tackle the uncertainty which uses discrete optimization techniques to explore all possible assignments. Though optimization techniques are very efficient in its own way, they are not scalable to big data. In this work, we consider causal inference testing with binary outcomes and propose computationally efficient algorithms that are scalable to large-scale observational studies. By leveraging the structure of the optimization model, we propose a robustness condition which further reduces the computational burden. We validate the efficiency of the proposed algorithms by testing the causal relation between Hospital Readmission Reduction Program (HRRP) and readmission to different hospital (non-index readmission) on the State of California Patient Discharge Database from 2010 to 2014. Our result shows that HRRP has a causal relation with the increase in non-index readmission and the proposed algorithms proved to be highly scalable in testing causal relations from large-scale observational studies. 




\textit{\textbf{Key words:}} Causal Inference,  Policy Evaluation,  Matching Method,  Robustness,  Public policy, Observational Study, Big Data.

\maketitle

%


\section*{Introduction} \label{intro}

Effective and evidence-based public policy decisions aim to manipulate one or many socio-economic variables and analyze their impact on the desired outcomes~\cite{nssah2006propensity}. The impact assessment is not associational but causal~\cite{nssah2006propensity,pearl2009causal} which requires an understanding of the counterfactual\textemdash the difference in outcomes with or without the presence of the policy~\cite{kleinberg2015prediction}. This is also true for any post policy evaluation~\cite{nssah2006propensity}. A policy maker may design multiple policies and calculate the causal quantities including the effect of the proposed policies on different recipient groups, effect over time, possible trade-offs between competing goals, and finally, choose the optimal policy~\cite{zajonc2012essays}. The gold standard approach for calculating those causal quantities is conducting a randomized experiment~\cite{rubin1974estimating,austin2011introduction,rosenbaum1983central,athey2017state}. In a randomized experiment, the experimenter can assign an observation to either a treatment or control group randomly; this randomness can avoid  bias and eliminate confounding effects of covariates and thus can achieve unbiased estimation of treatment effects. In this case, a possible association between treatment and outcome will imply causation. However, many studies in health care, social science, economics, and epidemiology cannot be designed as a randomized experiment due to legal or ethical reasons; randomization can also be impractical, time consuming, or very expensive. Hence, those experiments are performed on data that is collected as a natural process. Such experiments are called observational studies (also refereed as natural experiments or quasi-experiments)~\cite{rosenbaum2002observational} and can be implemented in a prospective (collecting sample data as natural observation over time) or retrospective (experimenting on already collected data) way. 

Making causal inference from an observational study lacks the experimental elements of randomization on all possible background covariates (the observed and unobserved characteristics of a sample unit)~\cite{stuart2010matching,stroup2000meta} and prone to bias, and systematic confounding on covariates. However, with proper understanding of the underlying process and careful control of non-randomized data, it is possible to make reasonable estimation of the causal effect~\cite{rubin1974estimating}. Researchers have been utilizing matching methods for identifying causality since the 1940s~\cite{stuart2010matching} and it is one of the most popular methods. It was used or referred in as many as 486,000 academic articles involving causal inference (see \nameref{s:1}). Matching methods examine the possibility of restoring or replicating properties of randomization based on the observed covariates~\cite{stuart2010matching}. In fact, matching attempts to retrieve the latent randomization within the observational data~\cite{hansen2004full}. Being true to it's name, matching methods aim to find a control group which is identical to the treatment group in terms of joint distribution of the observed covariates. As discussed by Stuart~\cite{stuart2010matching}, and Zubizarreta~\cite{zubizarreta2012using}, matching the empirical distribution of the covariates has several significant advantages. For example, matching forces the experimenter to closely examine the data, checking the common support on the covariates, and make the experimenter aware on the quality of inference. Even though the matching process can be complex, the outcome analysis is often done with simple methods~\cite{rosenbaum1985constructing}. For instance, the Rubin Causal Model (also known as Potential Outcome Framework) estimates the causal effect as the difference of expected outcomes between the control group and the treatment group~\cite{holland1986stat}. Due to it's simplistic architecture and many other attractive properties (see~\cite{stuart2010matching,zubizarreta2012using,morgan2006matching}), matching has been used to make policy decision or policy evaluation in health care~\cite{christakis2003health,akematsu2012measuring,kiil2012does,sari2015effects}, education~\cite{zubizarreta2017optimal,hong2006evaluating}, economics~\cite{dehejia1999causal}, law~\cite{epstein2005supreme}, and politics~\cite{herron2007assessing}. 

In this paper, we adopt a robust methodology recently proposed by Morucci \textit{et al.}~\cite{morucci2018hypothesis} and extend it to accommodate causal inference from big data observational studies. We show the efficiency of the proposed methods by evaluating the impact of Hospital Readmission Reduction Program (HRRP)~\cite{mcilvennan2015hospital} policy on Non-index readmission\textemdash readmission to a hospital that is different from the hospital which discharged the patient.

\subsection*{Motivation and contribution}
\subsubsection*{Motivation}

The objective of current one-to-one matching paradigm under potential outcome framework is to find pairs $(t,c)$ between samples $t$ from treatment group $\mathscr{T}$ and $c$ from control group $\mathscr{C}$. A pair $(t,c)$ is assigned in such a way that $t$ and $c$ are exactly same or similar on a specific, pre-determined set of covariates $\mathbf{X}$: $\left \{ (t,c): t \simeq c \vert \mathbf{X}; t \in \mathscr{T} \textrm{ and } c \in \mathscr{C} \right \}$. Over the years, researchers developed a wide array of algorithms to find such pairs, for example, Propensity Score matching~\cite{rosenbaum1985constructing}, Mahalanobis Distance matching~\cite{rosenbaum1985constructing}, Nearest Neighbour Greedy matching~\cite{austin2009some}, Coarsed Exact Matching~\cite{iacus2009cem}, and Genetic matching~\cite{diamond2013genetic} are among the most popular algorithms. All these algorithms (including those not listed here) disregard the outcomes $(Y^1_t, Y^0_c)$ of corresponding pairs $(t,c)$ in the assignment process. Though the matching process reduces bias in treatment effect estimation, disregarding the outcomes in the assignment process introduces a new source of uncertainty. If a sample $t\in \mathscr{T}$ has multiple possible pair assignments $\left \{ c_1, c_2, \cdots, c_n \right  \} \in \mathscr{C}$ and have similar covariate balance but different outcomes (i.e., $Y_t^1 - Y_{c_1}^0 \neq  Y_t^1 - Y_{c_2}^0 \neq \cdots \neq Y_t^1 - Y_{c_n}^0$), by assigning pairs without considering the outcomes, an experimenter can estimate multiple degrees of causal effect (one for each possible assignment). Similarly, a sample from control group $c\in \mathscr{C}$ can have multiple possible assignment options $\left \{ t_1, t_2, \cdots, t_n \right  \} \in \mathscr{T}$. A possible scenario is presented in Fig~\ref{fig1} where within each circle we have multiple pair assignment options with almost similar match quality but different outcomes (outcomes are presented as the size of the data points). In that case, different experimenters using different matching algorithms can get different pairs, hence, their causal effect estimates and conclusions on the experiment can be different. It is possible that two researchers having the exact same hypothesis and using the exact same data but, two different matching algorithms can achieve completely opposite results due to this uncertainty. This problem is exacerbating for studies involving big data as we may have more pair assignment options. Therefore, making policy decisions, in health care or any other field, using matching method that disregards uncertainty due to pair assignment can make a harmful impact on the society.
\begin{figure}[h!]
    \centering
    \caption{\textbf{Uncertainty due to multiple pair assignment options.} Shapes and Colors represent the treatment status and variations in size represent the difference in outcomes.}
    \includegraphics[scale = 0.5]{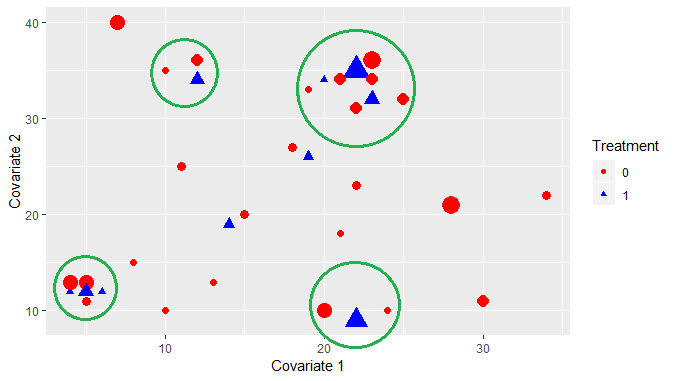}
    \label{fig1}
\end{figure}

On the other hand, in 2012, congress adopted HRRP under Patient Protection and Affordable Care Act (PPACA)~\cite{mcilvennan2015hospital} to increase quality of care and reduce hospital readmission rate. HRRP penalizes the hospital which discharged the patient (index hospital) if the patient with some specific conditions (i.e., Pneumonia, Acute Mayocardial Infraction (AMI), Congestive Heart Failure (CHF)) is readmitted within 30 days of discharge. The index hospital is always penalized even if the patient is readmitted to different hospital (non-index hospital)~\cite{mcilvennan2015hospital}. Though readmission to the index hospital is following a decreasing trend over post HRRP periods, non-index readmissions are increasing~\cite{chen2018reducing}. This increase in non-index readmission rate\textemdash approximately one fifth of all readmissions for medicare patients~\cite{chen2018reducing,chen2017hospital}\textemdash creates suspicion that hospitals are possibly discouraging patients for readmission to avoid penalties introduced by HRRP. Moreover, a recent study identified that non-index readmissions are associated with higher odds of in-hospital mortality and longer length of stay~\cite{burke2018influence}. Therefore, we aim to identify whether HRRP has a causal relation to the increase in non-index readmissions. Finding such causal relation involves analyzing large volume of health care data and matching method would be vulnerable to the uncertainty discussed above. The robust method proposed in~\cite{morucci2018hypothesis} to handle such uncertainty requires solving multiple Integer Programming (IP) models (a minimization and a maximization problem) iteratively. Using state-of-the-art integer programming solvers to solve those IP models for big data observational studies will be computationally expensive.

\subsubsection*{Contribution}
In this work, we extend the robust causal inference testing method proposed by Morucci \textit{et al.}~\cite{morucci2018hypothesis} to handle large scale observational studies with binary outcomes. To handle big data, first, we propose a robustness condition which identifies when a robust solution is possible and combines the maximization and minimization problems into a single problem. Second, we propose an efficient algorithm to calculate the test statistics for the robust condition. In addition, we propose two algorithms\textemdash one to solve the minimization problem and one to solve the maximization problem\textemdash for any condition which will show the degree of uncertainty for a selected number of matched pair. Finally, we implement the algorithms by testing causal effect of HRRP to non-index readmission using the State of California Patient Discharge Data and compare the computational efficiency with canonical IP solvers.    

\begin{remark}
\label{rem:1}
Please note, by ``Robust" we imply ``Robust to the choice of matching method": if $\mathscr{A}$ represents a set of all possible matching algorithms, a researcher choosing any algorithm $A_i \in \mathscr{A}$ and testing a hypothesis of causal effect will get the same result if she has chosen algorithm $A_{j\neq i} \in \mathscr{A}$. Also, we are considering matching as pre-processing and plan to achieve robust test result from a large-scale observational study for a given set of good matches $\mathscr{M}$ identified by any matching algorithm $A_i \in \mathscr{A}$. 
\end{remark}

\section*{Causal inference with matching method and robust test}
In the Rubin Causal Model, a sample unit $i$ from a set of observations $\left \{ 1, 2, \cdots , n \right \} \in \mathscr{S}$ can have two outcomes or responses. The response $Y_i^T$ is called treatment response when the unit $i$ receives certain treatment ($T=1$) and  control response when unit $i$ does not receive treatment ($T=0$). It is assumed that the treatment assignment to any unit does not interfere with the outcome of other units~\cite{rubin1978bayesian}. This assumption is commonly known as Stable Unit Treatment Value Assumption (SUTVA). Under this assumption, the treatment effect on a sample unit $i \in \mathscr{S}$ is calculated as $TE_i = Y_i^{1}-Y_i^{0}$. However, it is impossible to observe the counterfactual scenario for the same sample~\cite{holland1986stat}. Under a certain treatment regime $T\in \left\{ 0,1 \right  \}$ and identical conditions, we can only observe $Y_i^{T=1}$ or $Y_i^{T=0}$ for sample $i$: $Y_i = T_i Y_i^{1} + (1-T_i) Y_i^{0}$~\cite{rubin1974estimating,holland1986stat}. Therefore, we cannot directly measure the treatment effect $TE$ at an individual level. On the other hand, causal inference literature offers a statistical solution to this fundamental problem by taking expectation over the observation set $\mathscr{S}$, formally called \textit{Average Treatment Effect (ATE)}.

\begin{align}
\label{eq:1}
 ATE = E \left[ Y^{1} - Y^{0}|\mathbf{X}\right]   
\end{align}

The \textit{ATE} as defined in Eq~\ref{eq:1} provides the opportunity to divide $\mathscr{S}$ into the treatment group $\mathscr{T}$ when $T=1$ and control group $\mathscr{C}$ when $T=0$ such that $(\mathscr{T} \cup \mathscr{C}) = \mathscr{S}$  and work with their expectations. So, we can construct the \textit{ATE} as $E[Y^{1}|T=1] - E[Y^{0}|T=0]$ but, this form of \textit{ATE} implicitly assumes that the potential responses are independent of treatment assignment: $Y_i^1, Y_i^0 \perp T, \forall i \in \mathscr{S} $. Though this independence assumption holds in randomized experiments, in general, it does not hold for observational studies as the experimenter rarely has control on the treatment assignment process. This problem is solved by making an assumption known as Strong Ignorability~\cite{rosenbaum1983central}. Let $\mathbf{X} \in \mathcal{X}$ and $\mathbf{X} \in \mathbb{R}^{k}$ be the set of pre-treatment background variables (covariates) which characterizes the observations. The strong ignorability assumption states that the potential responses are independent of treatment assignment when conditioned on the covariates: $Y_i^1, Y_i^0 \perp T|\mathbf{X} $ and every unit $i \in \mathscr{S}$ has a positive probability to receiving treatment: $0 < Pr (T=1|\mathbf{X}=\mathbf{x}) < 1$. Another commonly used estimate of causal effect is \textit{Average Treatment Effect on Treated (ATT)} which is defined under slightly relaxed assumption ($Y_i^0 \perp T|\mathbf{X}$).

\begin{align}
\label{eq:2}
ATT = E [ (Y^{1} - Y^{0})|\mathbf{X}, T=1]    
\end{align}

Both of these estimates are prone to bias as the treatment assignment process is not random. In matching method,  an unbiased estimate of causal inference can be achieved if treatment unit $t \in \mathscr{T}$ is exactly matched with a control unit $c \in \mathscr{C}$ in terms of their covariate set $\mathbf{X} \in \mathcal{X}$~\cite{rosenbaum1983central}. However, in most of the applications, it is impossible to achieve exact matching~\cite{zubizarreta2012using,rosenbaum1983central,nikolaev2013balance,king2016propensity}. A wide variety of matching methods are employed to make $(t,c)$ pairs as similar as possible~\cite{zubizarreta2012using,zubizarreta2015stable,rosenbaum1983central} or finding a subset of control group samples $\mathcal{C}\subset \mathscr{C}$ which is similar to the treatment group samples $\mathcal{T} \subset\mathscr{T}$ in the joint distribution of the covariate set $\mathbf{X}$~\cite{nikolaev2013balance,iacus2009cem}. In this work, we consider one-to-one matching which aims to find a pair $(t,c) \subset (\mathscr{T}, \mathscr{C})$ that is matched (either exactly or by some user defined balance function) on a set of covariates $\mathbf{X} \subset \mathcal{X}$.

Before explaining the difference between classical method of causal inference~\cite{rubin1974estimating, rosenbaum1985constructing,holland1986stat} and robust causal inference testing approach~\cite{morucci2018hypothesis}, let us define the set of good match $\mathscr{M}$ and the pair assignment variables $a_{i,j}$.
\begin{definition}
\label{def:1}
\textup{(A set of Good Match)} A set of good match $\mathscr{M}$ includes treatment group samples $\mathcal{T} \subset \mathscr{T}$ and control group samples $\mathcal{C} \subset \mathscr{C}$ that satisfies certain covariate balance criteria defined under matching algorithm $A_i \in \mathscr{A}$.

\begin{align*}
\mathscr{M} := \left \{(t,c) \in (\mathcal{T} \times \mathcal{C}) : t \simeq c \vert \mathbf{X}    \right \}
\end{align*}
\end{definition}

\begin{definition}
\label{def:2}
\textup{(Pair Assignment Operator)} The Pair Assignment Operator is a binary assignment variable $a_{ij} \in \left \{ 0,1 \right \}$ where $a_{ij} = 1$ if sample $t_i \in \mathcal{T}$ is paired with a sample $c_j \in \mathcal{C}$ and the pair $(t_i,c_j) \in \mathscr{M}$; $a_{ij} = 0$ otherwise. 
\end{definition}
For a given set of possible matches $\mathscr{M}$, we can perform hypothesis test in the following form with the null hypothesis being no causal effect and alternative being the opposite.
\begin{equation}\label{eq:3}
    \mathbf{H}_0^{ATE} : \mathbb{E}[ Y^1-Y^0 \vert \mathbf{X} ] = 0
\end{equation}
\begin{equation}\label{eq:4}
    \mathbf{H}_0^{ATT} : \mathbb{E}[ Y^1-Y^0 \vert \mathbf{X}, T=1 ]= 0
\end{equation}
Under the classical approach of matching method, we can test these hypotheses first by defining a test statistic $\Lambda$, specifying an imbalance measure along with a tolerance limit on the imbalance. Then, we apply a matching algorithm $A_i \in \mathscr{A}$ to find the set of good match $\mathscr{M}$ which satisfy the imbalance limit, otherwise we tune the allowable imbalance limit to generate $\mathscr{M}$. Robust approach differs from the classical approach moving forward from here (see Fig~\ref{fig2}). The classical approach picks one (out of many) possible combination of pairs from $\mathscr{M}$ and conducts the hypothesis test wherein, the robust approach calculate the maximum and minimum value of the test statistic $(\Lambda_{max}, \Lambda_{min})$ and corresponding p-values to explore all possible assignment combinations within $\mathscr{M}$ which does not increase imbalance under Definition~\ref{def:1}. The test will be robust if both $\Lambda_{max}$ and $\Lambda_{min}$ produce same conclusion on the hypothesis. We formally define the Robust Test in Definition~\ref{def:3}.
\begin{figure}[!ht]
\centering
    \caption{\textbf{Comparison of matching for hypothesis testing under classical approach and robust approach~\cite{morucci2018hypothesis}.} Steps before covariate balance achievement remains same for each approach. In the remaining steps, Black arrows show the classical approach, and Blue arrows show the robust approach proposed in~\cite{morucci2018hypothesis}. }
   \includegraphics[width=1\textwidth]{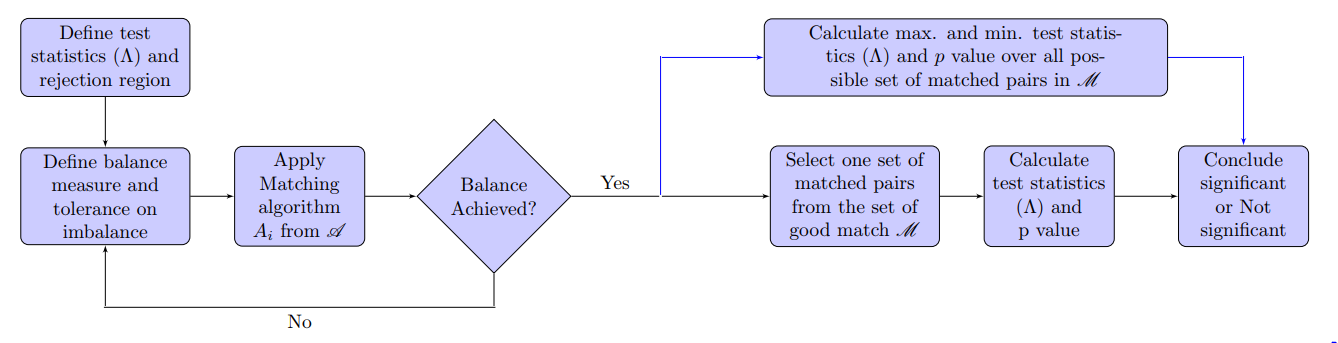}
    \label{fig2}
\end{figure}

\begin{definition}
\label{def:3}
\textup{(Robust Test)} Let $\alpha$ be the level of significance set for the hypothesis $\mathbb{H}_0$ and $(\Lambda_{max}, \Lambda_{min})$ are the test statistics calculated from $\mathscr{M}$, then, testing $\mathbf{H}_0$ is called $\alpha$-robust if max(\textup{p-value($\Lambda_{max}$), p-value($\Lambda_{min}$)}$)\leq \alpha$ or min(\textup{p-value($\Lambda_{max}$), p-value($\Lambda_{min}$)}$)> \alpha$. Testing $\mathbf{H}_0$ is called absolute-robust when p-value($\Lambda_{min}$) $=$ p-value($\Lambda_{max}$).
\end{definition}

Calculating the test statistic $\Lambda$ generates an integer programming model which is computationally expensive for large scale data. In the following section, we propose a robustness condition following the Robust Test definition which will allow us to calculate a $\Lambda_{robust} = \Lambda_{min} = \Lambda_{max}$ for absolute-robust test and we can avoid solving two integer programming problems. In this work, we are interested in testing the hypothesis stated in Eqs (\ref{eq:3},\ref{eq:4}) for binary outcomes: $Y\in \left \{ 0,1 \right \}$ with the McNemar's test~\cite{mcnemar1947note} as proposed in~\cite{morucci2018hypothesis}. From now on, we use the terms ``robust" and ``absolute-robust" interchangeably and when we use robust, we are implying to the absolute-robust test defined in Definition~\ref{def:3}.

\section*{Robust McNemar's test }
McNemar's test is the ideal candidate for testing hypothesis in Eqs (\ref{eq:3},\ref{eq:4}) as it deals with on-to-one matched pairs. It operates on a $2 \times 2$ contingency table (see Table~\ref{tab:conf}) and the test statistics under null hypothesis assumes that the marginal proportions are homogeneous. Among the four types of matched pairs, we are mainly interested in the discordant pairs $B = \sum_{i\in \mathscr{T}} \sum_{j \in \mathscr{C}} a_{ij}Y^0_j (1-Y^1_i)$ and $C = \sum_{i\in \mathscr{T}} \sum_{j \in \mathscr{C}} a_{ij}Y^1_i (1-Y^0_j)$ where $a_{i,j}$ is the pair assignment operator defined in Definition~\ref{def:2}. Here, $B$ counts the number of pairs where treatment units has outcomes 0: $Y^1 = 0$ and control units has outcomes 1: $Y^0 = 1$ and $C$ counts the discordant pairs where $Y^1 = 1$ and $Y^0 = 0$. Under the assumption of having at least 1 discordant pair: $B+C \geq 1$ we will use the test statistic $\Lambda$ as defined in Eq (\ref{eq:5}) to test both hypotheses.
\begin{equation}\label{eq:5}
    \Lambda = \frac{B-C-1}{\sqrt{B+C}}
\end{equation}

\begin{table}
\centering
\caption{\textbf{Contingency table of the outcomes of treatment and control observations.}}
\label{tab:conf}
\begin{tabular}{cccc}
                          & \multicolumn{3}{c}{\hspace{2 cm} Treatment}                                                                                                                                            \\
                          &                          & Yes $\left(Y^1 =1\right)$                                                                   & No $\left(Y^1 =0\right)$                                                                   \\ \cline{3-4} 
                          & \multicolumn{1}{l|}{Yes $\left(Y^0 =1\right)$} & \multicolumn{1}{l|}{\cellcolor[HTML]{34FF34}{\color[HTML]{000000} \hspace{0.75 cm} A}} & \multicolumn{1}{l|}{\cellcolor[HTML]{9698ED}{\color[HTML]{000000} \hspace{0.75 cm} B}} \\ \cline{3-4} 
\multirow{-2}{*}{Control} & \multicolumn{1}{l|}{No $\left(Y^0 =0\right)$  }  & \multicolumn{1}{l|}{\cellcolor[HTML]{FD6864}{\color[HTML]{000000} \hspace{0.75 cm} C}} & \multicolumn{1}{l|}{\cellcolor[HTML]{34FF34}{\color[HTML]{000000} \hspace{0.75 cm} D}} \\ \cline{3-4} 
\end{tabular}
\end{table}

Morucci \textit{et al.}~\cite{morucci2018hypothesis} proposed the following integer programming model that explores all possible assignment options and calculate maximum and minimum possible test statistics $\Lambda_{max}$ and $\Lambda_{min}$, respectively.

$$\textrm{Maximize/Minimize}_a \quad \Lambda(\textbf{a}) = \frac{B-C-1}{\sqrt{B+C}}$$

Subject to:
\begin{align}
& \sum_{i\in \mathscr{T}} \sum_{j \in \mathscr{C}} a_{ij}Y^0_j (1-Y^1_i) = B \\
& \sum_{i\in \mathscr{T}} \sum_{j \in \mathscr{C}} a_{ij}Y^1_i (1-Y^0_j) = C \\
& B+C = m \quad \textrm{(Total number of discordant pairs)} \\
& \sum_{i \in \mathscr{T}} a_{ij} \leq 1 \quad \forall j \quad \textrm{(Choose at most one treatment observation)} \\
& \sum_{j \in \mathscr{C}} a_{ij} \leq 1 \quad \forall i \quad \textrm{(Choose at most one control observation)}
\end{align}
$$\textrm{Additional user-defined covariate balance constraints to find }  \mathscr{M} $$
\begin{align}
a_{ij} \in \left \{0,1\right \}   
\end{align}

As it is shown in Table~\ref{tab:conf}, $B$ is the total number of untied responses when $Y^1_i = 0$ is matched with $Y^0_j = 1$. Similarly, $C$ is total number of untied responses when $Y^1_i = 1$ is matched with $Y^0_j = 0$. Therefore, both $B,C \in \mathbb{R}^+$. Under this definition of $B$ and $C$, we provide the following propositions on the objective function of the robust McNemar's test and its optimal values.

\begin{prop}
\label{proposition:1}
The objective function $\Lambda(\textbf{a})$, has the following properties:
\begin{enumerate}[(i)]
\item For any $C > 0$, $\Lambda(\textbf{a})$ is strictly increasing in $B$ for $B \in \mathbb{R}^+$
\item For any $B \geq 0$, $\Lambda(\textbf{a})$ is monotonically decreasing in $C$ for $C\geq 1$ and strictly decreasing for $C > 1$
\end{enumerate}
\end{prop}

\begin{proof}
Let $C > 0$, then for any $B \in \mathbb{R}^+$, we have
\begin{align}
\frac{\partial \Lambda(\textbf{a})}{\partial B} = \frac{B+3C+1}{2(B+C)^{3/2}} > 0
\end{align}
which implies $\Lambda(\textbf{a})$ is strictly increasing in $B$ for a fixed $C$. Similarly, let $B \geq 0$, then for any $C\geq 1$ we have,
\begin{align}
\frac{\partial \Lambda(\textbf{a})}{\partial C} = \frac{-3B-C+1}{2(B+C)^{3/2}} \leq 0
\end{align}
this proves the claims of Proposition~\ref{proposition:1}.
\end{proof}

Before further discussion on $\Lambda(\textbf{a})$ and the optimality conditions, we introduce following notations and definitions of maximum untied responses, for both $B$ and $C$.
\\
$|Y^1_i = 1|$ is the number of treatment units in the matched set with positive outcome\\
$|Y^1_i = 0|$ is the number of treatment units in the matched set with negative outcome\\
$|Y^0_j = 1|$ is the number of control units in the matched set with positive outcome\\
$|Y^0_j = 0|$ is the number of control units in the matched set with negative outcome

\begin{definition} 
\label{def:4}
\textup{(Maximum type one discordant pair)} $B_{max}$ in the maximum number of possible pairs between $Y^1_i \in \mathscr{T}$ and $Y^0_j \in \mathscr{C}$ where the treated observation has negative (``No") outcome but the untreated (control) observation has positive (``Yes") outcome, i.e., 
$$B_{max} = \min\left \{|Y^0_j = 1|, |Y^1_i = 0| \right \}$$
\end{definition}

\begin{definition} 
\label{def:5}
\textup{(Maximum type two discordant pair)} $C_{max}$ in the maximum number of possible pairs between $Y^1_i \in \mathscr{T}$ and $Y^0_j \in \mathscr{C}$ where the treated observation has positive (``Yes") outcome but the untreated (control) observation has negative (``No") outcome, i.e.,
$$C_{max} = \textrm{min}\left \{|Y^1_i = 1|, |Y^0_j = 0| \right \}$$
\end{definition}

\noindent For a fix value of $m$, the McNemar's test model becomes linear and the objective functions become,
\begin{align}
\Lambda(\textbf{a}) = \frac{1}{\sqrt{m}} (B-C-1)   
\end{align}

Using the property of $\Lambda(\textbf{a})$ explained in Proposition~\ref{proposition:1}, we can find the optimal solution.


\begin{prop}
\label{proposition:2}
Let $C \geq 1$ and denote $m$ as the total number of discordant pairs, then the optimal pair $(C^*,B^*)$ is given by:
\begin{align*}
& \textbf{min:} \quad (C^*,B^*) =\left\{\begin{matrix}
\left(C_{max}, m-C_{max}\right) & \quad \textrm{if} \quad m> C_{max}\\ 
\left(m, 0\right) & \quad \textrm{if} \quad m < C_{max}
\end{matrix}\right. \\
& \textbf{max:} \quad (C^*,B^*) =\left\{\begin{matrix}
\left(m-B_{max}, B_{max}\right) & \quad \textrm{if} \quad m> B_{max}\\ 
\left(0, m \right) & \quad \textrm{if} \quad m < B_{max}
\end{matrix}\right.
\end{align*}
\end{prop}

\begin{proof}
From Proposition~\ref{proposition:1}, we know that $\Lambda(\textbf{a})$ is monotonically decreasing in C when $C \geq 1$. Therefore, in the minimization problem, assignment will be made to maximize $C$ until we are about to violate  constraint $B+C =m$. When the total number of discordant pairs is set to $m > C_{max}$, $C$ will take the value of $C_{max}$ and $B$ will take the value of $m-C_{max}$ just to satisfy the total number of discordant pair constraints and the solution will be optimal. If $m < C_{max}$, the new $C= m$ and the minimum value will be achieved at $C=m$ and $B=0$. 
\\
Similarly, from Proposition~\ref{proposition:1}, we know that $\Lambda(\textbf{a})$ is strictly increasing in $B$ for any $B \in \mathbb{R}^+$. So, in the maximization problem, pair assignment will be made to maximize $B$ within the feasible region. When the total number of discordant pair is set to $m > B_{max}$, at optimal solution, $B$ will take the value of $B_{max}$ and $C$ will take the value of $m-B_{max}$ just to stay in the feasible region. When $m < B_{max}$, the $B$ will take the value $m$ and the maximum value will be achieved at $B=m$ and $C=0$.
\end{proof}
 
\begin{prop}
\label{proposition:3}
 For the linear model, an absolute-robust estimate will be achieved if and only if the total number of discordant pair  $m=B_{max} + C_{max}$.
\end{prop}

\begin{proof}
\noindent According to the proposed approach to the causal inference estimate, an absolute-robust estimate is achieved when $\Lambda(\textbf{a})_{max}$ and $\Lambda(\textbf{a})_{min}$ is equal. For the McNemar's test model, the model becomes infeasible when $m$ is set to $m > B_{max}+C_{max}$ as we can only have $B_{max}+C_{max}$ number of total untied responses. So feasible range of $m$ is: $0 < m \leq (B_{max}+C_{max})$.

To prove the Proposition~\ref{proposition:3}, we first set $m$ to it's maximum value $B_{max}+C_{max}$. Using Proposition~\ref{proposition:2}, in this case, the optimal solution for the $\Lambda(\textbf{a})_{max}$ problem is: $B = B_{max}$, $C = m-B_{max} = C_{max}$ and the optimal solution for $\Lambda(\textbf{a})_{min}$ problem is: $C = C_{max}$, $B=m-C_{max} = B_{max}$. So, for $m=B_{max}+C_{max}$ case, we get $\Lambda(\textbf{a})_{max} = \Lambda(\textbf{a})_{min}$ and the solution is absolute-robust.  

Conversely, $m$ can take any integer value in the range $0 < m < (B_{max}+C_{max})$ which can lead to the following six cases. For each of the cases, we will find the optimal solution using Proposition~\ref{proposition:2}. 
\begin{itemize}
    \item $0 \leq B_{max} \leq C_{max} \leq m < (B_{max}+C_{max})$: The optimal solution for the minimization problem is $C = C_{max}, B = m - C_{max}$ and the maximization problem is $C = m-B_{max}, B = B_{max}$.
    \item $0 \leq B_{max} < m \leq C_{max} < (B_{max}+C_{max})$: The optimal solution for the minimization problem is $C = m, B = 0$ and the maximization problem is $C = m-B_{max}, B = B_{max}$.
    \item $0 \leq C_{max} \leq B_{max} \leq m < (B_{max}+C_{max})$: The optimal solution for the minimization problem is $C = C_{max}, B = m - C_{max}$ and the maximization problem is $C = m-B_{max}, B = B_{max}$.
    \item $0 \leq C_{max} \leq m \leq B_{max} < (B_{max}+C_{max})$: The optimal solution for the minimization problem is $C = C_{max}, B = m - C_{max}$ and the maximization problem is $C = 0, B = m$.
    \item $0 < m \leq C_{max} \leq B_{max} < (B_{max}+C_{max})$: The optimal solution for the minimization problem is $C = m, B = 0$ and the maximization problem is $C = 0, B = m$.
    \item $0 < m \leq B_{max} \leq C_{max} < (B_{max}+C_{max})$: The optimal solution for the minimization problem is $C = m, B = 0$ and the maximization problem is $C = 0, B = m$.
\end{itemize}
For all of the above six cases, $\Lambda(\textbf{a})_{max} \neq \Lambda(\textbf{a})_{min}$, hence, the solution is not absolute-robust.

\noindent Therefore, the total number of discordant pairs $m$ have to be $B_{max}+C_{max}$ to get an absolute-robust estimate.
\end{proof}

As we can see from the Proposition~\ref{proposition:2}, the optimization problem has become a counting problem and can be solved efficiently for big data. However, the optimal solution under Proposition~\ref{proposition:2} disregards the assignment constraints Eqs (9,10) and addition user-defined constraints. To find the optimal solution using the result from Proposition~\ref{proposition:2} that is feasible, we take a two-step approach. At the first step, we handle the user-defined constraints to find a good set of match $\mathscr{M}$ as a pre-processing step. We can use any off-the-shelf matching algorithm for that purpose or define a separate pair assignment model with different covariate balance measure to find $\mathscr{M}$. At the second step, we partition the set of good match $\mathscr{M}$ into $\mathcal{P}$ partitions such that within a partition $p \in \left \{1,2,\cdots, \mathcal{P} \right \}$, any treatment unit $t$ can be matched with any control unit $c$. A formal definition of a partition is provided below.

\begin{definition}
\label{def:6}
\textup{(Partition of $\mathscr{M}$)} $p \subset \mathscr{M}$ is a partition if any treatment unit $t \in \left \{ 1,2,\cdots, \mathcal{N}^p_t \right \}$ is a good match to any control unit $c \in \left \{ 1,2,\cdots, \mathcal{N}^p_c \right \}$ and $(t,c) \in \mathscr{M}$. Reverse is also true.
\end{definition}

Construction of partitions under Definition~\ref{def:6} ensures that only good matches are considered for assignment. In addition, Definition~\ref{def:4} calculates $B_{max}$ by pairing negative outcomes of treatment units and positive outcomes of control units which inherently satisfies the pair assignment constraints Eqs (9,10). Similarly, we calculate $C_{max}$ by assigning a pair between samples with positive treatment outcomes and negative control outcomes. Therefore, none of the treatment or control unit is used more than once in the pair assignment process which satisfies the pair assignment constraints Eqs (9,10).   

Now, using the above mentioned results, we propose Algorithm~\ref{alg:1} which identifies the robustness condition and corresponding test statistics $\Lambda(\mathbf{a})_{robust}$.

\begin{algorithm} 
\caption{: Test statistic $\Lambda(\mathbf{a})_{robust}$ at robustness condition} 
\label{alg:1} 
\begin{algorithmic} 
    \REQUIRE Vector of outcomes $(Y^1,Y^0)^1, (Y^1,Y^0)^2, \cdots, (Y^1,Y^0)^\mathcal{P}$ 
    \STATE $B_{max} \gets 0$
    \STATE $C_{max} \gets 0$
    \FOR{$p=1:\mathcal{P}$}
        \STATE $B^p \gets min(\vert Y^1 = 0 \vert, \vert Y^0 =1\vert)^p$
        \STATE $C^p \gets min(\vert Y^1 = 1 \vert, \vert Y^0 =0\vert)^p$
        \STATE $B_{max} \gets B_{max} + B^p$
        \STATE $C_{max} \gets C_{max} + C^p$
    \ENDFOR
    \RETURN \[\Lambda(\mathbf{a})_{robust} = \frac{B_{max}-C_{max}-1}{\sqrt{B_{max}+C_{max}}}\]
\end{algorithmic}
\end{algorithm}

By the sketch of the Algorithm~\ref{alg:1}, it seems like we are only matching the discordant pairs and ignoring the other possible pair assignments in the data, which is not true. In Proposition~\ref{proposition:4}, we show that we match maximum possible pairs.  
\begin{prop}
\label{proposition:4}
Algorithm~\ref{alg:1} ensures that the maximum possible pairs $(t,c)$ are matched in $\mathscr{M}$.
\end{prop}

\begin{proof}
To prove this Proposition, we only need to show that in any partition $p$, Algorithm~\ref{alg:1} matches maximum possible pairs. Then, we can sum the maximum pair assignments across the partitions to achieve maximum possible pairs $(t,c)$ assignment in $\mathscr{M}$.
 
Lets consider a partition $p$ where $\mathcal{N}^p_t$ denotes the number of treatment samples and $\mathcal{N}^p_c$ denotes the number of control samples. Hence, the maximum number of pairs we can assign in $p$ is $min(\mathcal{N}^p_t,\mathcal{N}^p_c)$. We will use $\mathcal{N}^{p+}_.$ to represent the number of samples with positive outcomes ($Y^. = 1$) and $\mathcal{N}^{p-}_.$ to represent the number of samples with negative outcomes ($Y^. = 0$). After assigning the discordant pairs ($B^p_{max}$) and ($C^p_{max}$) as we did in Algorithm~\ref{alg:1}, we are left with  $(\mathcal{N}^{p+}_t - C^p_{max}) + (\mathcal{N}^{p-}_t - B^p_{max})$ treatment samples and $(\mathcal{N}^{p+}_c - B^p_{max}) + (\mathcal{N}^{p-}_c - C^p_{max})$ control samples. Now, we can assign the remaining treatment and control samples into the other two types of pairs $A$ and $D$ to their limit: 
\begin{align*}
  & A^p_{max} = min\left((\mathcal{N}^{p+}_t - C^p_{max}), (\mathcal{N}^{p+}_c - B^p_{max})\right) \\
  & D^p_{max} = min\left((\mathcal{N}^{p-}_t - B^p_{max}), (\mathcal{N}^{p-}_c - C^p_{max})\right)
\end{align*}

It is trivial to show that the for partition $p$,
\begin{align*}
 min(\mathcal{N}^p_t,\mathcal{N}^p_c) =B^p_{max} +C^p_{max}+D^p_{max}+A^p_{max}   
\end{align*}
An example of maximum pair assignment is provided in Fig~\ref{fig4} where treatment outcomes ($t$) are sorted in descending order and control outcomes ($c$) are sorted in ascending order. In the left panel, we can have $min(\mathcal{N}^p_t,\mathcal{N}^p_c) = 5$ pairs at maximum. After assigning $B_{max} =2$ and $C_{max} =2$ according to Algorithm~\ref{alg:1}, we can assign only one pair to $D_{max}$ and $A_{max}=0$. Therefore, we achieve the maximum number of pair assignments. We follow the similar procedure in the middle and right panels.
\end{proof}

\begin{figure}[!ht]
\centering
    
    \caption{\textbf{Example of maximum pair assignments between treatment and control group.} $t$ represents the treatment group and $c$ represents the control group. An arrow connects a treatment unit with a control unit which forms a pair.}
    \includegraphics[width=0.9\textwidth]{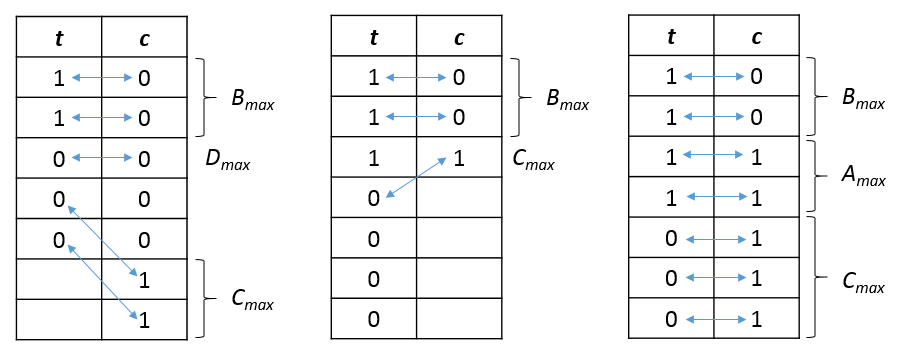}
    \label{fig4}
\end{figure}

In the Algorithm~\ref{alg:1}, we calculate the test statistics for robust condition which an experimenter can use to find the corresponding p-value and compare with a pre-defined level of significance $\alpha$ to make conclusion on the hypothesis of no causal relation. Decisions made in this process will be free of uncertainty and robust to the choice of matching algorithms. If different experimenters perform matching on same data using different matching algorithms but follow the above mentioned procedure, all of their conclusions will be exactly the same. 

Regarding the computational complexity arises due to big data, our proposed algorithm only involves counting elements in vectors and few algebraic operations. The counting processing can be done with the summation of vectors as we are dealing with only binary outcomes: summation implies the total number of positive outcomes and we can calculate the negative outcomes by subtracting it from the size of the vector. In addition, we only need to solve the problem once\textemdash at robustness condition. Therefore, the proposed algorithm will be highly efficient for big data. 

While the Algorithm~\ref{alg:1} directly calculates test statistics at robustness condition, a researcher might be interested in exploring the degree of uncertainty in the causal inference test. She may want to see how the uncertainty changes towards the robust estimate with respect to the number of discordant pairs matched. For this purpose, we propose the following two algorithms (\ref{alg:2}, \ref{alg:3}) following the result of Proposition~\ref{proposition:2}.

\begin{algorithm} 
\caption{: Maximizing the test statistics $\Lambda(\mathbf{a})$} 
\label{alg:2} 
\begin{algorithmic} 
    \REQUIRE Vector of outcomes $(Y^1,Y^0)^1, \cdots, (Y^1,Y^0)^\mathcal{P}$ and increment in $m$: $I$
    \STATE $B, C \gets 0$
    \WHILE{$m \leq B_{max} + C_{max}$}
        \FOR{$p=1:\mathcal{P}$}
            \IF{$m < B_{max}$}
                \STATE $C^p, B^p \gets 0,m$
            \ELSE
                \STATE $C^p, B^p \gets m-B_{max}, B_{max}$
            \ENDIF
            \STATE $B \gets B+B^p$
            \STATE $C \gets C+C^p$
            \IF{$(B+C) \geq m$}
                \STATE break
            \ENDIF
        \ENDFOR
        \STATE $m \gets m + I$
    \ENDWHILE
    \RETURN \[\Lambda(\mathbf{a})_{max} = \frac{B-C-1}{\sqrt{B+C}}\]
\end{algorithmic}
\end{algorithm}

\begin{algorithm} 
\caption{: Minimizing the test statistics $\Lambda(\mathbf{a})$} 
\label{alg:3} 
\begin{algorithmic} 
    \REQUIRE Vector of outcomes $(Y^1,Y^0)^1, \cdots, (Y^1,Y^0)^\mathcal{P}$ and increment in $m$: $I$
    \STATE $B, C \gets 0$
    \WHILE{$m \leq B_{max} + C_{max}$}
        \FOR{$p=1:\mathcal{P}$}
            \IF{$m < B_{max}$}
                \STATE $C^p, B^p \gets m,0$
            \ELSE
                \STATE $C^p, B^p \gets C_{max}, m-C_{max}$
            \ENDIF
            \STATE $B \gets B+B^p$
            \STATE $C \gets C+C^p$
            \IF{$(B+C) \geq m$}
                \STATE break
            \ENDIF
        \ENDFOR
        \STATE $m \gets m + I$
    \ENDWHILE
    \RETURN \[\Lambda(\mathbf{a})_{min} = \frac{B-C-1}{\sqrt{B+C}}\]
\end{algorithmic}
\end{algorithm}

\section*{Numerical experiment}
In this section, we present the efficiency of the proposed algorithms with data from the State of California Patient Discharge Database and answer an interesting hypothesis on the effectiveness of HRRP implemented in October 1, 2012. 

Readmission rate is considered an important measure of hospital or care quality. To increase the care quality and hold hospitals accountable, US Congress introduces HRRP under PPACA in 2012~\cite{mcilvennan2015hospital}. The most important feature of this program is that the index hospital (the hospital that discharged the patient) is penalized if patients with Pneumonia, Congestive Heart Failure (CHF), and Acute Mayocardial Infraction (AMI) are readmitted (to the index hospital or any other hospital) within 30 days of discharge. During the post HRRP period, the overall rate of readmission is decreasing which the proponents of HRRP are attributing as the success of the policy. However, in this period, readmission to different hospitals (non-index readmission) has been increasing~\cite{chen2018reducing,chen2017hospital}. Non-index readmissions are considered as a discontinuity of care and found to be associated with longer lengths of stay and higher in-hospital mortality rates. There is a chance that hospitals are possibly discouraging patients for readmission to avoid penalties introduced by HRRP. To check the cause behind the increase in non-index readmission post HRRP periods, we make the following hypothesis and test it with the proposed algorithms with the level of significance $\alpha =0.05$. 

\textit{H\textsubscript{0}: HRRP has no causal relation with the increase in non-index readmission}

\textit{H\textsubscript{1}: HRRP has a positive causal relation with the increase in non-index readmission}

\subsection*{Data description and covariate balance}
In this research, we primarily used the patient discharge data between 2010 to 2014 from California. We collected this nonpublic data from the California Office of Statewide Health Planning and Development (OSHPD) which collects in-patient data from California licensed hospitals. Each patient in this data set is presented with a unique identifier which is used to determine if a patient is readmitted in the future. In addition, the data set also contains patient level information such as, ICD-9 codes for clinical diagnosis, comorbidity index, age, gender, discharge destination, patients' zip code, and insurance information. The patients are tracked with unique identifiers to find out if they were readmitted within 30 days of discharge and when a readmission is found, we identify the destination hospital of that readmission. Then, a binary variable is created with 0 if the patient is readmitted to the same hospital or 1 if different hospital. To test the hypothesis, we will use this variable as our outcome: $Y = 1$ if readmitted to a different hospital or $Y=0$ if readmitted to the same hospital.   

Moreover, the OSHPD data set is merged with publicly available data from Centers for Medicare and Medicaid Services, American Association Annual Hospital Survey and the Area resource file. From these additional data sources, we get important hospital level information like, teaching status (membership status of the Council of Teaching Hospitals), ownership type (public, non-profit, investor owned), hospital size based on number of beds (small: below 100 beds, medium: 101 to 399 beds, and large: 400 and above beds), hospital locations (rural, metro). We also include patient household incomes which we consider as the median income of patients' residence zip codes. We divide the data into two sets: before and after October 1, 2012, the implementation date of HRRP. The treatment here is the implementation of HRRP, all the readmissions between February 1, 2010 and September 30, 2012 is considered as control group $\mathscr{C}$ (treatment $T=0$) and all the readmission from October 1, 2012 to November 30, 2014 is considered as treatment group $\mathscr{T}$ (treatment $T=1$). To capture any potential readmission within 30 days period of an index discharge, admissions before February 1, 2010 and beyond November 30, 2014 are excluded. A descriptive view of readmitted patients' characteristics is presented in Table~\ref{Table:1}.

\begin{table}
\small
\begin{adjustwidth}{-1.3 in}{0in} 
\caption{\textbf{Characteristics of readmitted patients in the State of California Patient Discharge Database from 2010 to 2014.}}
\label{Table:1}
\begin{tabular}{lccccc}
\thickhline
\textbf{Variable}   & \textbf{All Readmission} & \textbf{Index Hospital} & \textbf{Non-index Hospital} & \textbf{Before HRRP}   & \textbf{After HRRP}    \\ \thickhline
\textbf{Readmitted Patients}         & 90553           & 67341             & 23212          & 53353  & 37200  \\ \thickhline
\textbf{Demographic Characteristics} &                 &                   &                       &               &               \\ \hline
\textbf{Age}                         &                 &                   &                       &               &               \\ \hline
\quad 0-20                        & 635 (0.70)      & 505 (0.75)        & 130 (0.56)            & 427 (0.8)     & 208 (0.56)    \\
\quad 21-30                       & 1073 (1.18)     & 717 (1.06)        & 356 (1.53)            & 566 (1.06)    & 507 (1.36)    \\
\quad 31-40                       & 2186 (2.41)     & 1471 (2.18)       & 715 (3.08)            & 1269 (2.38)   & 917 (2.47)    \\
\quad 41-50                       & 6336 (7.00)     & 4196 (6.23)       & 2140 (9.22)           & 3714 (6.96)   & 2622 (7.05)   \\
\quad 51-65                       & 21018 (23.21)   & 14470 (21.49)     & 6548 (28.21)          & 11950 (22.4)  & 9068 (24.38)  \\
\quad 65 and above                & 59305 (65.49)   & 45982 (68.28)     & 13323 (57.4)          & 35427 (66.4)  & 23878 (64.19) \\ \hline
\textbf{Gender}                     &                 &                   &                       &               &               \\ \hline
\quad Female                      & 45124 (49.80)   & 34240 (50.80)     & 10884 (46.9)          & 27049 (59.94) & 18075 (40.06) \\
\quad Male                        & 45429 (50.20)   & 33101 (49.20)     & 12328 (53.1)          & 26304 (57.9)  & 19125 (42.1)  \\ \hline
\textbf{Household Income}           &                 &                   &                       &               &               \\ \hline
\quad Quartile 1                  & 22428 (24.77)   & 15640 (23.23)     & 6788 (29.24)          & 13108 (24.57) & 9320 (25.05)  \\
\quad Quartile 2                  & 22629 (24.99)   & 16633 (24.70)     & 5996 (25.83)          & 13331 (24.99) & 9298 (24.99)  \\
\quad Quartile 3                  & 22450 (24.79)   & 17160 (25.48)     & 5290 (22.79)          & 13165 (24.68) & 9285 (24.96)  \\
\quad Quartile 4                  & 23046 (25.45)   & 17908 (26.59)     & 5138 (22.14)          & 13749 (25.77) & 9297 (24.99)  \\ \hline
\textbf{Clinical Characteristics}    &                 &                   &                       &               &               \\ \hline
\textbf{Primary Diagnosis}           &                 &                   &                       &               &               \\ \hline
\quad CHF                         & 50151 (55.40)   & 37404 (55.50)     & 12747 (54.9)          & 29351 (55.01) & 20800 (55.91) \\
\quad AMI                         & 11917 (13.20)   & 8148 (12.10)      & 3769 (16.2)           & 6865 (12.87)  & 5052 (13.58)  \\
\quad Pneumonia                   & 28485 (31.40)   & 21789 (32.40)     & 6696 (28.9)           & 17137 (32.12) & 11348 (30.51) \\ \hline
\textbf{Charlson Comorbidity Index}  &                 &                   &                       &               &               \\ \hline
\quad Low (0-2)                   & 35394 (39.09)   & 25884 (38.44)     & 9510 (40.97)          & 21282 (39.89) & 14112 (37.94) \\
\quad Medium (3-6)                & 51301 (56.65)   & 38454 (57.10)     & 12847 (55.35)         & 29850 (55.95) & 21451 (57.66) \\
\quad Medium High (7-10)          & 3396 (3.75      & 2628 (3.90)       & 768 (3.31)            & 1944 (3.64)   & 1452 (3.9)    \\
\quad High (10 and above)         & 462 (0.51)      & 375 (0.56)        & 87 (0.37)             & 277 (0.52)    & 185 (0.5)     \\ \hline
\textbf{Hospital Characteristics}    &                 &                   &                       &               &               \\ \hline
\textbf{Teaching Status}             &                 &                   &                       &               &               \\ \hline
\quad Teaching Hospital          & 10261 (11.30)   & 7706 (11.40)      & 2555 (11)             & 5882 (11.02)  & 4379 (11.77)  \\
\quad Non-teaching Hospital       & 80272 (88.70)   & 59635 (88.50)     & 20657 (89)            & 47471 (88.98) & 32821 (88.23) \\ \hline
\textbf{Ownership Type}              &                 &                   &                       &               &               \\ \hline
\quad Non-profit Hospital         & 58592 (64.70)   & 45210 (67.10)     & 13382 (57.6)          & 34252 (64.2)  & 24340 (65.43) \\
\quad Investor Hospital           & 17902 (19.80)   & 11389 (16.90)     & 6513 (28.1)           & 10839 (20.32) & 7063 (18.99)  \\
\quad Public Hospital             & 14059 (15.50)   & 10742 (16.00)     & 3317 (14.3)           & 8262 (15.49)  & 5797 (15.58)  \\ \hline
\textbf{Hospital Size}               &                 &                   &                       &               &               \\ \hline
\quad Small (below 100 beds)      & 4982 (5.50)     & 3453 (5.10)       & 1529 (6.6)            & 2979 (5.58)   & 2003 (5.38)   \\
\quad Medium (100-399 beds)       & 61167 (67.60)   & 45149 (67.10)     & 16018 (69)            & 36314 (68.06) & 24853 (66.81) \\
\quad Large (400 and above beds)  & 24404 (26.90)   & 18739 (27.80)     & 5665 (24.4)           & 14060 (26.35) & 10344 (27.81)\\ \hline
\textbf{Hospital Location}               &                 &                   &                       &               &               \\ \hline
\quad Rural      & 2426 (2.68)     & 1809 (2.69)       & 617 (2.66)            & 1348 (2.53)   & 1078 (2.90)   \\
\quad Metro       & 88127 (97.32)   & 65532 (97.31)     & 22595 (97.34)            & 52005 (97.47) & 36122 (97.10) \\\thickhline
\end{tabular}
 \justify{The entries in each cell is presented as Number of patients ``N (\%)" form. From February 1, 2010 to September 30, 2012 is considered ``Before HRRP" period. From October 1, 2012 to December 31, 2014 is considered ``After HRRP" period. CHF-Congestive Heart Failure, AMI-Acute Mayocardial Infraction.} 

\label{Tab:1}
\end{adjustwidth}
\end{table}

We match the patients based on the following covariates: Age, Gender, Primary diagnosis, Household Income, Charlson Comorbidity Index, Hospital Location, Hospitals' Teaching Status, Hospitals' Ownership Status, and Hospital Size. We divide the covariates into two groups: 1) discrete and 2) continuous. The discrete covariates (i.e., Gender, Primary Diagnosis, Hospital Location, Hospitals' Teaching Status, Hospitals' Ownership Status, and Hospital Size)  are matched exactly. The continuous covariates (i.e., Age, Household Income, and Charlson Comorbidity Index) are first, divided into categories as shown in Table~\ref{Table:1}, then, the categories are matched exactly. This matching strategy resulted in 1822 partitions of data; within a partition, any treatment sample can be matched with any control sample. Though this matching approach seems ad hoc in nature, it is very similar to the well known method called Coarsed Exact matching (CEM)~\cite{iacus2009cem} with 1822 bins. Traditionally, CEM is implemented with much lesser number bins due to the lack of common support between treatment and control groups but a higher number of bins make finer covariate balance~\cite{iacus2012causal,iacus2011multivariate} which is the objective of any matching method. However, to implement the proposed algorithms, an experimenter is not limited to CEM or the matching method we used. Given a good set of matchec created under Definition~\ref{def:1}, we can always create the partitions under Definition~\ref{def:6}.

\subsection*{Experiment and result}
To test the hypothesis \textit{H\textsubscript{0}}, first, we perform the matching operations in R~\cite{ihaka1996r} to get matched sets. Then, using the matched sets of data, we calculate the test statistic $\Lambda(\mathbf{a})_{max}$ and $\Lambda(\mathbf{a})_{min}$ using the 1) optimization model with an Integer Programming solver and 2) with the proposed algorithms. The integer programming model is implemented in AMPL~\cite{fourer1987ampl} and solved with the commercial solver CPLEX~\cite{cplex2009v12}. We implement the Algorithm~\ref{alg:1}, \ref{alg:2}, and~\ref{alg:3} in R~\cite{ihaka1996r}. All the experiments are performed in a Dell Precision workstation with 64 GB RAM, Intel(R) Xeon(R) CPU E5-2670 v3 processor running at 2.30 GHz. 

Table~\ref{table2} shows the comparison of solutions obtained using an optimization model with CPLEX iterating over different values of discordant pairs ($m$) and proposed algorithm at robustness condition. The range of p-value achievable corresponding to the test statistics $\Lambda(\mathbf{a})_{max}$ and $\Lambda(\mathbf{a})_{min}$ is presented in Fig~\ref{fig5}. The proposed Algorithm~\ref{alg:1} directly identifies the robustness condition which is $B_{max} = 12082$ and $C_{max} = 9448$ and calculates the test statistics $\Lambda(\mathbf{a})_{robust}$. The computation time required by the proposed algorithms are very insignificant compared to the time required by the optimization model. A major implication of the robust McNemar's test is that if the same experiment is conducted with as many as 19,000 discordant pairs, we can achieve any p-value between 0 to 0.23 (see Fig~\ref{fig5}); some experimenter might rejected the hypothesis and some might fail to reject the null hypothesis. Both the experimenters, in this case, are right but the conclusion differed due to the fact that they choose different pairs. Any policy decision made using the matching method without considering this uncertainty has a possibility to fail. 

\begin{table}[!ht]
\begin{adjustwidth}{-0.0in}{0in} 
\centering
\caption{
{\bf Test statistic $\Lambda(\mathbf{a})$ calculated using optimization model and algorithm~\ref{alg:1}.}}
\begin{tabular}{l l l l l l}
\thickhline
 & \multicolumn{3}{l}{\bf Optimization Model} & \multicolumn{2}{l}{\bf Algorithm~\ref{alg:1}}\\ \thickhline
$m$ & $\Lambda(\mathbf{a})_{min}$ & $\Lambda(\mathbf{a})_{max}$ & CPU time & Robustness Condition& CPU time\\ \thickhline
$50$ & -7.21 & 6.93 & 918.69 &  & \\ 
$100$ & -10.10 & 9.90 & 982.23 &  & \\ 
$300$ & -17.38 & 17.26 & 1203.68 &  & \\ 
$500$ & -22.41 & 22.32 & 2037.37 &  & \\
$800$ & -28.32 & 28.25 & 2204.52 &  & \\
$1000$ & -31.65 & 61.60 & 2218.27 &  & \\
$5000$ & -70.72 & 70.69 & 2563.32 &  &  \\
$10000$ & -88.97 & 99.99 & 2934.47 &   &  \\
$15000$ & -31.82 & 74.82 & 2386.53 &  & \\
$20000$ & 7.80 & 21.83 & 2659.94 &  & \\
$21000$ & 14.51 & 21.83 & 2640.60 &  & \\
$21500$ & 17.75 & 18.16 & 2219.23 &  & \\
$21530$ & 17.94 & 17.94 & 3246.64 & \quad  $17.94^*$  & $3^*$ \\\thickhline
\end{tabular}
\begin{flushleft} The optimization model is solved iteratively over different values of discordant pairs ($m$) until a robust solution is reached. $^*$Algorithm~\ref{alg:1} identifies the robustness condition ($B_{max} = 12082$ and $C_{max} = 9448$) and calculates the test statistic for that condition only. CPU times are presented in seconds: time required to solve both minimization and maximization problem. 
\end{flushleft}
\label{table2}
\end{adjustwidth}
\end{table}

\begin{figure}[H]
\caption{{\bf The range of p-value achievable for different number of discordant pairs $m$.}
The p-values are calculated using the test statistics presented in Table~\ref{table2}. The red line represents minimum possible p-value (corresponding to $\Lambda(\mathbf{a})_{max}$) and the blue line represents the maximum possible p-value (corresponding to $\Lambda(\mathbf{a})_{min}$).}
    \includegraphics[width=1\textwidth]{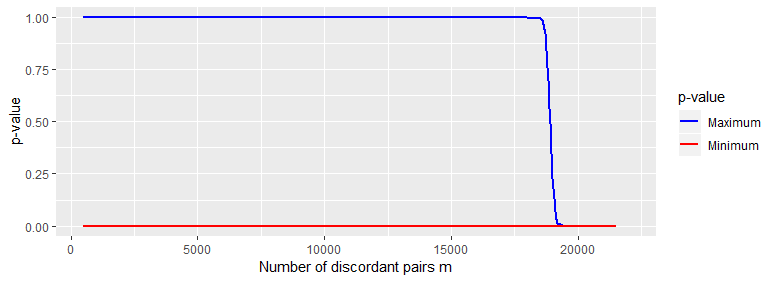}
\label{fig5}
\end{figure}
In regards to the hypothesis we made at the beginning of this section, we can make a conclusion on the hypothesis by using the p-value calculated at the robustness condition or the result from Table~\ref{table2} and corresponding p-values from Fig~\ref{fig5}. We can see that, both the maximum and minimum p-value $< \alpha$ when we match more than 20,000 discordant pairs. Therefore, we can reject the null hypothesis of no causal effect and conclude that HRRP is a cause for increase in the non-index readmissions. This result points out that not only the readmission rate but also the non-index readmission rate should be considered as a measure of health care quality and health care policy makers should take another look at the implication of HRRP policy.

\section*{Conclusion}
Any policy decision or evaluation requires identifying the causal relation between policy alternatives and potential outcomes. Matching methods have become very popular in identifying such causal relations. However, in one-to-one matching, when we have multiple pair assignment options, matching method is vulnerable to uncertainty as the pair construction process does not consider outcomes. In this paper, we consider the integer programming model for robust causal inference testing approach with binary outcomes proposed by Morucci \textit{et al.}~\cite{morucci2018hypothesis} and develop scalable algorithms that can be used for large-scale observational studies. We identify a robustness condition which combines the maximization and minimization problem proposed in~\cite{morucci2018hypothesis}. Instead of solving two computationally expensive integer programming models iteratively by increasing the number of discordant pairs until a robust estimate is achieved, we convert the problems into counting problems through a series of propositions. In addition, the proposed Algorithm~\ref{alg:1} only solves one problem instead of two separate problems and it is computationally efficient. Numerical experiment conducted on State of California Patient Discharge Database shows that the proposed algorithms are highly scalable. The numerical experiment shows an interesting result on a well appreciated health care policy\textemdash HRRP\textemdash proposed under PPACA in 2012 to improve health care quality. We identify that HRRP is a cause to the increase of non-index readmission which is associated to higher in-hospital mortality rate and longer length of stay. Though the numerical experiment is performed with around 100,000 samples, the algorithms proposed in this paper can handle observational studies with millions of samples efficiently without further modification. In the future, we plan to develop similar robust causal inference testing algorithms with continuous outcomes for large-scale observational studies. 



\section*{Acknowledgments}
We thank Mr. Tasnim Ibn Faiz, PhD candidate, MIE, Northeastern University for his help and advice on programming in AMPL. We also thank Mr. Md Mahmudul Hasan, PhD candidate, MIE, Northeastern University for sharing his knowledge on the OSHPD database and supporting us throughout the data analysis process.

\bibliographystyle{unsrt}
\bibliography{robust}

\end{document}